\documentclass[cameraready]{jcmsi}%

\usepackage{graphicx}
\usepackage{latexsym}
\usepackage{cite}
\usepackage{amsmath}
\usepackage{comment,color}

\usepackage{txfonts}%%!!

\makeatletter%
\input{ot1txtt.fd}
\makeatother%

\newtheorem{prop}{Proposition}
\newtheorem{proof}{Proof}
\newtheorem{problem}{Problem}

\Vol{4}
\No{1}
\Month{1}
\Year{2011}

\title{Digital Cancelation of Self-Interference for Single-Frequency Full-Duplex Relay Stations via Sampled-Data Control}

\begin{document}

\AUTHOR{%
  \author{Hampei Sasahara}{Kyoto},
  \author{Masaaki Nagahara}{Kyoto},
  \author{Kazunori Hayashi}{Kyoto},
  \author{Yutaka Yamamoto}{Kyoto}
}

\AFFILIATE{%
  \affiliate{Graduate School of Informatics, 
    Kyoto University, 
    Japan}{Kyoto}
%%  \affiliate{Affiliation}{ALabel}
}

\email{nagahara@ieee.org
(corresponding author)
}

\begin{abstract}%
In this article, we propose sampled-data design of digital filters that cancel the
continuous-time effect of coupling waves in a single-frequency full-duplex relay station.
In this study, we model a relay station as a continuous-time
system while conventional researches treat it as a discrete-time system. 
For a continuous-time model, we
propose digital
feedback canceler based on
the sampled-data $H^\infty$ control theory to cancel coupling waves taking intersample behavior into account.
We also propose robust control against 
unknown multipath interference.
Simulation results are shown to illustrate the effectiveness of the proposed method.
\end{abstract}

\begin{keywords}%
wireless communication, coupling wave cancelation, sampled-data control, $H^{\infty}$ optimization, relay station.

\begin{comment}
wireless communication, single-frequency full-duplex relay, sampled-data $H^{\infty}$ control
\end{comment}
\end{keywords}

\received{0}{00}{2011}
\revised{0}{00}{2011}

\maketitle

%
% Section 1
%
\section{Introduction}
\label{sec:intro}
In wireless communications,
relay stations have been used to relay radio signals
between radio stations that cannot directly communicate
with each other due to the signal attenuation.
More recently, relay stations are used to achieve a certain spatial diversity
called cooperative diversity  to cope with fading channels \cite{Laneman}.
On the other hand,
it is important to efficiently utilize the scarce bandwidth
due to the limitation of frequency resources
\cite{Cov}, while conventional relay stations commonly use
different wireless resources, such as frequency, time and code,
for their reception and transmission of the signals.
For this reason, a single-frequency full-duplex relay station,
in which signals with the same carrier frequency are received
and transmitted simultaneously, is considered as one of key technologies
in the fifth generation (5G) mobile communications systems \cite{2020beyond}.
In order to realize such full-duplex relay stations,
{\em self-interference} caused by coupling waves is the key issue \cite{Jain}.

Fig.~\ref{coupling} illustrates self-interference by coupling waves.
In this figure, radio signals with carrier frequency $f$
are transmitted from the base station (denoted by BS).
One terminal (denoted by T1) directly receives the signal from the base station,
but the other terminal (denoted by T2) 
is so far from the base station that they cannot communicate directly.
Therefore, a relay station (denoted by RS) is attached between them
to relay radio signals.
Then, radio signals with the same carrier frequency $f$
are transmitted from RS to T2, but also they
are fed back to the receiving antenna
directly or through reflection objects.
As a result, self-interference is caused in the relay station,
which may deteriorate the quality of communication
and, even worse, may destabilize the closed-loop system.

To tackle with the issue of self-interference cancelation, 
many methods have been proposed 
for single-frequency full-duplex systems.
Analog cancelation has been proposed in \cite{Rad,Knox},
in which analog devices are used for canceling coupling waves.
Since coupling wave paths are physically analog systems
and there is no quantization problem,
this design is theoretically the most ideal except for implementation issues.
%A difficulty in self-interference cancelation is due to the fact that
%power of the interference signals is far and away higher than power of desired signals.
%It is a merit of analog cancelation that there is no necessary to consider quantization error.
On the other hand,
\textit{digital cancelation} 
has attracted increasing attention,
in which the interference is subtracted in the digital domain
by using digital signal processing techniques
\cite{sakai2006simple,Dua,Gol,Chun,Snow,Haya13,Sen}.
Digital cancelers benefit easy implementation on digital devices
in exchange for the response between sampling instants.
In addition, spatial domain techniques, called antenna cancelation, 
has been also proposed in \cite{Jain,Kho},
in which they try to reduce the interference by arranging
antenna placement.
See \cite{Jain,Dua} for details.

	\begin{figure}[t]
		\centering
		\scalebox{0.55}{\includegraphics{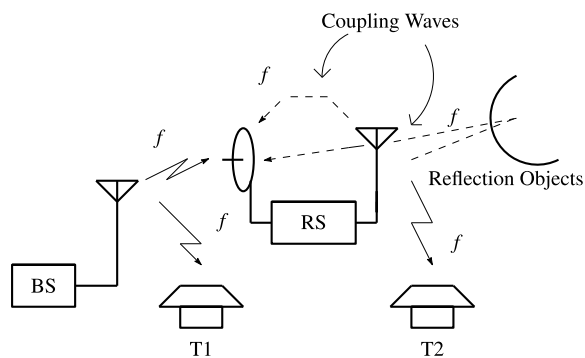}}
		\caption{Self-interference}
		\label{coupling}
	\end{figure}

For the problem of self-interference,
a pre-nulling method \cite{Chun} and
adaptive methods \cite{sakai2006simple,Haya13}
have been proposed 
to cancel the effect of coupling waves.
%a least mean square (LMS) adaptive filters
%\cite{Sak}, and adaptive array antennas 
%\cite{Nog}.
In these studies, a relay station is modeled by
a discrete-time system, and the performance
is optimized in the discrete-time domain.
However, radio waves are in nature continuous-time
signals and hence the performance should be
discussed in the continuous-time domain.
In other words, one should take account of {\em intersample behavior}
for coupling wave cancelation.

In theory, if the signals are completely band-limited below the
Nyquist frequency,
then the intersample behavior can be restored 
from the sampled-data in principle \cite{Shannon},
and the discrete-time domain approaches might work well.
However, the assumption of perfect band limitedness is
hardly satisfied in real signals since
\begin{enumerate}
\item real baseband signals are not fully band-limited,
\item pulse-shaping filters, such as raised-cosine filters, do not act perfectly,
\item and the nonlinearity in electric circuits
adds frequency components beyond the Nyquist frequency.
\end{enumerate}
One might think that if the sampling frequency is fast enough,
then the assumption is almost satisfied and there is no problem.
But this is not true; firstly, the sampling frequency cannot be arbitrarily
increased in real systems, and secondly, even though the sampling
is quite fast, intersample oscillations may happen
in feedback systems \cite[~Sect.~7]{Yamamoto99-1}.

To solve the problem mentioned above, 
we propose a new design method for coupling wave cancelation
based on the {\em sampled-data control theory}
\cite{Chen95,Yamamoto99-1}.
We model the transmitted radio signals and coupling waves
as continuous-time signals, and optimize the 
worst case continuous-time error due to coupling waves
by a {\em digital} canceler.
This is formulated as a sampled-data $H^\infty$ optimal control problem,
which can be solved via the fast-sampling fast-hold (FSFH) method
\cite{Kel,Yamamoto1999729}.
%In this study, we consider a feedback canceler,
%for which we place a digital controller in the feedback loop
%for stabilizing the feedback system as well as canceling
%the self-interference.
%This is formulated as a standard sampled-data $H^\infty$
%control problem except for the time delay in the feedback loop,
%which can be solved via FSFH as well.
We also propose robust feedback cancelers that
can take account of uncertainties in
coupling wave path characteristic
such as unknown multipath interference due to,
for example, large structures that reflect radio waves, or
the change of weather conditions \cite{Tak}.
Design examples are shown to illustrate the proposed methods.

The present manuscript expands on our recent conference contributions
\cite{SSHRsci,sasaharaSICE14} by incorporating
robust feedback control into the formulation.

The remainder of this article is organized as follows.
In Section \ref{sec:rs}, we derive a mathematical model of a relay station considered in this study.
In Section \ref{sec:fb}, we propose sampled-data $H^\infty$ control for cancelation of self-interference.
Here we also discuss robust control against uncertainty in the delay time.
In Section \ref{sec:sim}, simulation results are shown to illustrate the effectiveness of the proposed method.
In Section \ref{sec:conc}, we offer concluding remarks.

\subsection*{Notation}
Throughout this article, we use the following notation.
We denote by $L^2$ the Lebesgue space consisting of all square integrable real functions 
on $[0, \infty)$ endowed with $L^2$ norm $\|\cdot\|_2$.
The symbol $t$ denotes the argument of time, $s$ the argument of Laplace transform and $z$ the argument of $Z$ transform.
These symbols are used to indicate whether a signal or a system is of continuous-time or discrete-time.
The operator $e^{-Ls}$ with nonnegative real number $L$ denotes the continuous-time delay operator
with delay time $L$.
For a matrix $A$, $\overline{\sigma}(A)$ denotes the maximum singular value of $A$.

%
% Section 2
%
\section{Relay Station Model}
\label{sec:rs}
In this section, we provide a mathematical model of a relay station
with self-interference phenomenon.
	\begin{figure}[t]
	\includegraphics[width = 85mm]{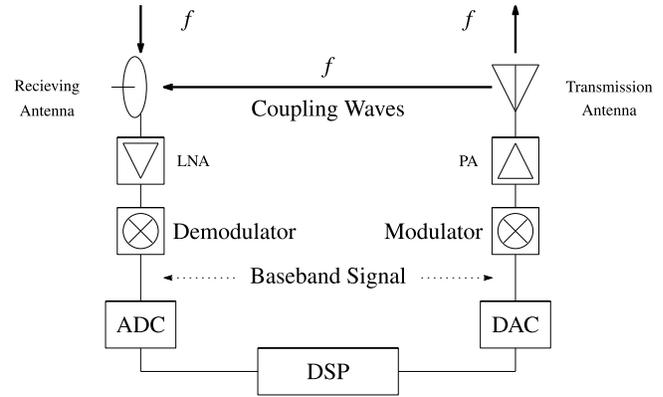}
	\caption{Relay Station}
	\label{Relay Station}
	\end{figure}

Fig.~\ref{Relay Station} depicts a single-frequency full-duplex relay station 
implemented with a digital canceler \cite{Knox}.
A radio wave with carrier frequency $f$
from a base station is accepted at the receiving antenna
and amplified by the low noise amplifier (LNA).
Then, the received signal is demodulated to a baseband signal 
by the demodulator,
and converted to a digital signal
by the analog-to-digital converter (ADC).
The obtained digital signal is then processed by the digital signal processor (DSP)
into another digital signal, which is converted to an analog signal by the
digital-to-analog converter (DAC).
Finally, the analog signal is modulated to a radio wave with
carrier frequency $f$, amplified by the power amplifier (PA)
and transmitted by the transmission antenna.
A problem here is that the transmitted signal will again reach the receiving antenna.
This is called coupling wave and causes self-interference,
which deteriorates the communication quality.

	\begin{figure}[t]
	\includegraphics[width = 85mm]{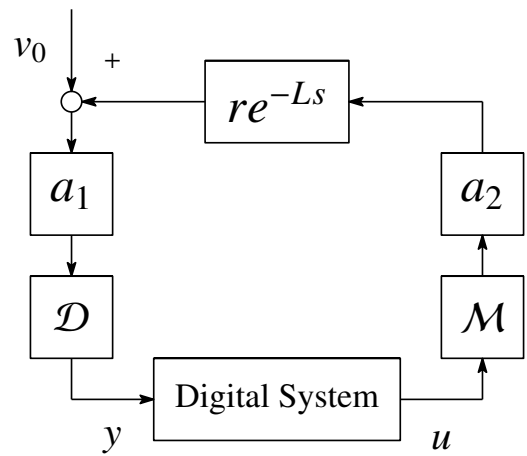}
	\caption{Simple Block Diagram of Relay Station}
	\label{Relay Station2}
	\end{figure}

Fig.~\ref{Relay Station2} shows a simplified block diagram of the relay station.
In Fig.~\ref{Relay Station}, we model
LNA and PA in Fig.~\ref{Relay Station} as static gains, $a_1$ and $a_2$, respectively.
The modulator is denoted by
${\mathcal M}$ and the demodulator by ${\mathcal D}$.
We assume that the coupling wave channel is a flat fading channel,
that is, all frequency components of a signal through this channel experience
the same magnitude fading.
Then the channel can be treated as an all-pass system.
In this study, we adopt a delay system, $re^{-Ls}$, as a channel model,
where $r>0$ is the attenuation rate and $L>0$ is a delay time.
The block named ``Digital System'' includes ADC, DSP and DAC in Fig.~\ref{Relay Station}.

In this article, we consider the quadrature amplitude modulation (QAM), 
which is used widely in digital communication systems, as a modulation method.
QAM transforms a transmission signal into two orthogonal carrier waves,
that is a sine wave and a cosine wave.
We assume the transmission signal $u(t)$ is given by
\begin{Meqnarray}
	u(t) := \sum_{k} g(t-kh) \left[
	\begin{array}{c}
	u_k^I \\
	u_k^Q \\
	\end{array}
	\right].
\end{Meqnarray}
In this expression, $g(t)$ is a general pulse-shaping function,
$h$ is the sampling period, and $u_k^I,u_k^Q$ denote respectively the in-phase
and the quadrature components of a transmission symbol.
We assume that the support of the Fourier transform $G(j\omega)$ of $g(t)$
is finite and the bandwidth is much less than $4 \pi f$.
In other words, there exists a frequency $f_g $ ($0<f_g\ll f$) such that 
$|G(j\omega)|=0$ for any $\omega \notin (-2\pi f_g,2\pi f_g)$.
Then the modulated signal $\tilde{u}(t)$ can be written as \cite[~Chap.~2]{HayComm}
\begin{Meqnarray}
	\tilde{u}(t) &=& {\mathcal M} u(t)\nonumber\\
	     &=& \sum_{k} g(t-kh) ( u_k^I \cos 2\pi ft- u_k^Q \sin 2\pi ft).
\end{Meqnarray}

On the other hand, 
the demodulation operator ${\mathcal D}$ is a linear operator
satisfying ${\mathcal D}{\mathcal M}=1$ \cite{HayComm}.
Fig.~\ref{Demodulator} shows the block diagram of ${\mathcal D}$.
In this block diagram, $H_{\rm id}(j\omega)$ is the ideal low-pass filter with 
cut-off frequency $f_c$ satisfying 
\begin{Meqnarray}
 H_{\rm id}(j\omega) = \begin{cases}
 	1, & \text{~if~} \omega<2\pi f_c,\\
	0, & \text{~otherwise}.
	\end{cases}
\end{Meqnarray}
The cut-off frequency is chosen to satisfy
$f_g \ll f_c \ll f$.
By the linearity of ${\mathcal D}$,
we obtain an equivalent block diagram shown in Fig.~\ref{Relay Station3} to Fig.~\ref{Relay Station2}.
Here
\begin{Meqnarray}
\tilde{u}(t-L)
 &=& \sum_k g(t-L-kh) \biggl\{ \bigl(u_k^I \cos (2\pi fL) \nonumber\\
 &+& u_k^Q \sin (2\pi fL)\bigr) \cos (2\pi ft) \nonumber\\
 &+& \bigl(u_k^I \sin (2\pi fL) - u_k^Q \cos (2\pi fL)\bigr) \sin (2\pi ft)\biggr\}.
\end{Meqnarray}
Thus, we have
\begin{Meqnarray}
& &\cos (2\pi ft) \cdot \tilde{u}(t-L)\nonumber\\
&~& =\frac{1}{2} \sum_k g(t-L-kh) \biggl\{ u_k^I \cos (2\pi fL)
+ u_k^Q \sin (2 \pi fL) \nonumber\\
&~& +(u_k^I \cos (2\pi fL) + u_k^Q \sin (2\pi fL)) \cos (4 \pi ft) \nonumber\\
&~& +(u_k^I \sin (2\pi fL) - u_k^Q \cos (2\pi fL)) \sin (4 \pi ft) \Big\}.
\end{Meqnarray}
From this, we have
\begin{Meqnarray}
 & & 2H_{\rm id}[\cos(2\pi ft) \cdot \tilde{u}(t-L)]\nonumber\\
 &~& = \sum_k g(t-L-kh) \bigl\{u_k^I \cos(2\pi f L) + u^Q_k \sin(2\pi f L)\bigr\}.
\end{Meqnarray}
In the same way, we have
\begin{Meqnarray}
 & & 2H_{\rm id}[-\sin(2\pi ft) \cdot \tilde{u}(t-L)]\nonumber\\
 &~& = \sum_k g(t-L-kh) \bigl\{-u_k^I \sin(2\pi f L) + u^Q_k \cos(2\pi f L)\bigr\}.
\end{Meqnarray}
Finally, we have the following relation:
\begin{Meqnarray}
	u_L(t) &=& {\mathcal D}\left( a_1ra_2 \tilde{u}(t-L) \right)\nonumber \\
	&=& \alpha A_L u(t-L),
\end{Meqnarray}
where $\alpha := a_1a_2r$ and
\begin{Meqnarray}
	A_L := \left[
	\begin{array}{cc}
	\cos (2\pi fL) & \sin (2\pi fL) \\
	-\sin (2\pi fL) & \cos (2\pi fL) \\	
	\end{array}
	\right].
\end{Meqnarray}

	\begin{figure}[t]
	\centering
	\includegraphics[width = 70mm]{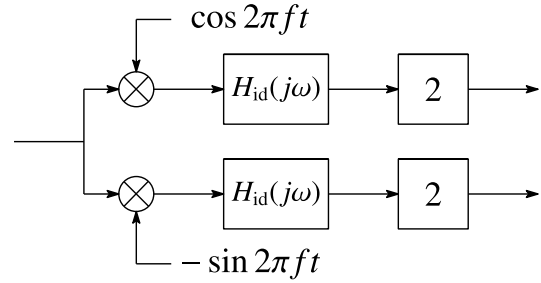}
	\caption{Structure of a demodulator}
	\label{Demodulator}
	\end{figure}

	\begin{figure}[t]
	\includegraphics[width = 85mm]{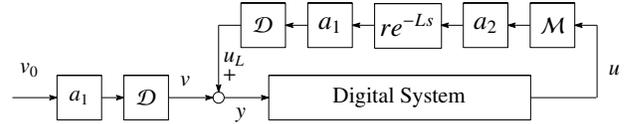}
	\caption{Equivalent Block Diagram of Relay Station}
	\label{Relay Station3}
	\end{figure}

Finally we obtain a relay station model depicted in Fig.~\ref{Relay Station Model}.
By this figure, we can see that the relay station with self-interference is a \emph{feedback} system.
In practice, the gain of  PA in Fig.~\ref{Relay Station} is very high
(e.g.,  $a_2=1000$) and the loop gain becomes much larger than $1$,
and hence we should discuss the \emph{stability} as well as self-interference cancelation.
To achieve these requirements, we design the digital controller in the digital system,
which is precisely shown in Fig.~\ref{Digital Factors}.
	\begin{figure}[t]
	\includegraphics[width = 85mm]{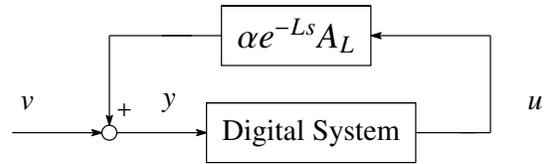}
	\caption{Relay Station Model}
	\label{Relay Station Model}
	\end{figure}
	\begin{figure}[t]
	\centering
	\includegraphics[width = 80mm]{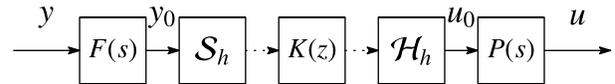}
	\caption{Digital System}
	\label{Digital Factors}
	\end{figure}

In Fig.~\ref{Digital Factors}, ADC in Fig.~\ref{Relay Station} is modeled by an
anti-aliasing analog filter $F(s)$ with an ideal sampler ${\mathcal S}_h$ with sampling period $h > 0$,
defined by
\begin{Meqnarray}
	\begin{array}{c}
	{\mathcal S}_h : \{ y_0(t) \} \mapsto \{ y_d[n] \} : y_d[n] = y_0(nh), \\
     n = 0,1,2,\ldots. \\
	\end{array}
\end{Meqnarray}
For the DSP block in Fig.~\ref{Relay Station}, we assume a digital filter denoted by $K(z)$,
which we design for self-interference cancelation.
DAC in Fig.~\ref{Relay Station} is modeled by 
a zero-order hold, ${\mathcal H}_h$, 
defined by
\begin{Meqnarray}
	\begin{array}{c}
	{\mathcal H}_h : \{u_d[n]\} \mapsto \{u_0(t)\}:u_0(t)=u_d[n], \\
	t \in [nh, (n+1)h), n=0,1,2,\ldots,\\
	\end{array}
\end{Meqnarray}
and a post analog low-pass filter denoted by $P(s)$.
We assume that $F(s)$ and  $P(s)$ are proper, stable and real-rational transfer function 
matrices.
Note that a strictly proper function is normally used for $F(s)$
and it is included in the assumption.

%
% Section 3
%
\section{Feedback Control}
\label{sec:fb}
Fig.~\ref{Feedback Canceler} shows the block diagram of the feedback control system
of the relay station.
	\begin{figure}[t]
	\centering
	\includegraphics[width = 80mm]{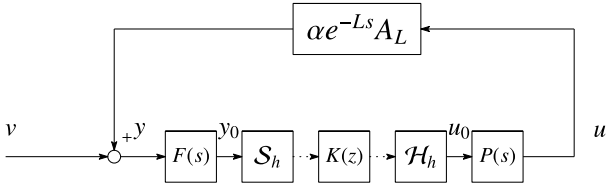}
	\caption{Feedback Canceler}
	\label{Feedback Canceler}
	\end{figure}
For this system, we find the digital controller, $K(z)$, that stabilizes the feedback system
and also minimize the effect of self-interference, $z:=v-u$, for any $v$.
To obtain a reasonable solution, we restrict the input continuous-time signal $v$ to the following set
\begin{Meqnarray}
	WL^2 := \{v=Ww:w \in L^2, \|w \|_{2} = 1\},
\end{Meqnarray}
where $W$ is a continuous-time LTI system with real-rational, stable, and strictly proper transfer function $W(s)$.
Under this assumption, we first solve a \emph{nominal} control problem where all system parameters are
previously known. Then we propose a \emph{robust} controller design against uncertainty 
in the coupling wave paths.

%
% Subsection 3.1
%
\subsection{Nominal Controller Design}
Here we consider the nominal controller design problem formulated as follows:
\begin{problem}
Find the digital controller (canceler) $K(z)$ that stabilizes the feedback system
in Fig.~\ref{Feedback Canceler} and uniformly minimizes the $L^2$ norm of the error $z = v-u$ for any $v \in WL^2$.
\end{problem}

This problem is reducible to a standard sampled-data $H^{\infty}$ control problem \cite{Chen95,Yamamoto99-1}.
To see this, let us consider the block diagram shown in Fig.~\ref{Feedback Canceler2}.
Let $T_{zw}$ be the system from $w$ to $z$.
Then we have
\begin{Meqnarray}
	z = v-u = T_{zw}w
\end{Meqnarray}
and hence uniformly minimizing $\| z\|_{2}$ for any $v \in WL^2$ is equivalent to minimizing the $H^{\infty}$ norm of $T_{zw}$,
\begin{Meqnarray}
	\| T_{zw}\|_{\infty} = \sup_{w \in L^2,~\|w \|_{2} = 1} \|T_{zw} w \|_{2}.
\end{Meqnarray}

Let $\Sigma (s)$ be a generalized plant given by
\begin{Meqnarray}
	\Sigma (s) = \left[
	\begin{array}{cc}
	W(s) & -P(s) \\
	F(s)W(s) & \alpha e^{-Ls}A_LF(s)P(s) \\
	\end{array}
	\right].
\end{Meqnarray}
By using this, we have
\begin{Meqnarray}
	T_{zw}(s) = {\mathcal F}(\Sigma (s), {\mathcal H}_h K(z) {\mathcal S}_h),
\end{Meqnarray}
where ${\mathcal F}$ denotes the linear-fractional transformation (LFT) \cite{Chen95}.
Fig.~\ref{Ge_Feedback} shows the block diagram of this LFT.
Then our problem is to find a digital controller $K(z)$ that minimizes $\|T_{zw}\|_{\infty}$.
This is a standard sampled-data $H^{\infty}$ control problem, and can be efficiently solved via FSFH approximation \cite{Kel,Nagahara13-2,Yamamoto1999729}.

\begin{figure}[t]
\includegraphics[width = 85mm]{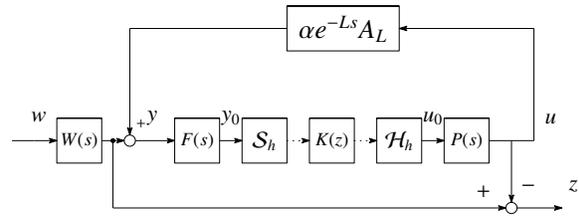}
\caption{Block Diagram for Feedback Canceler Design}
\label{Feedback Canceler2}
\end{figure}

\begin{figure}[t]
\centering
\includegraphics[width = 70mm]{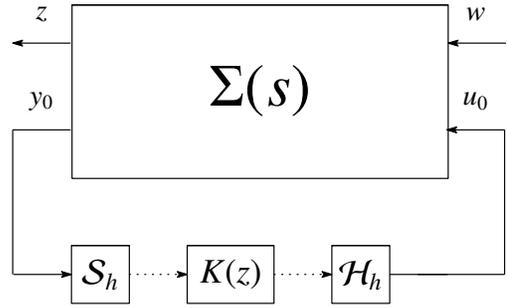}
\caption{LFT $T_{zw}={\mathcal F}(\Sigma, {\mathcal H}_hK{\mathcal S}_h)$}
\label{Ge_Feedback}
\end{figure}

Note that if there exists a controller $K(z)$ that minimizes $\|T_{zw}\|_{\infty}$, then the feedback system is stable and the effect of self-interference $z=v-u$ is bounded by the $H^{\infty}$ norm.
We summarize this as a proposition.

\begin{prop}
Assume $\|T_{zw}\|_{\infty} \leq \gamma$ with $\gamma >0$.
Then the feedback system shown in Fig.~\ref{Feedback Canceler} is stable, and for any $v \in WL^2$ we have $\|v-u\|_{2} \leq \gamma$.
\end{prop}
\begin{proof}
First, if the feedback system is unstable, then the $H^{\infty}$ norm becomes unbounded.
Next, for $v \in WL^2$ there exists $w \in L^2$ such that $v = Ww$ and $\| w \|_{2} = 1$.
Then, inequality $\| T_{zw} \|_{\infty} \leq \gamma $ gives
\begin{Meqnarray}
	\|v-u\|_{2} = \|T_{zw}w\|_{2} \leq \|T_{zw}\|_{\infty} \|w\|_{2} \leq \gamma.
\end{Meqnarray}
\hfill $\Box$
\end{proof}

%
% subsection 3.2
%

\subsection{Robust Controller Design against Multipath Interference}
\label{sec:rob}
In practice, the characteristic of the coupling wave channel changes
due to, for example,
large structures that reflect radio waves.
In this situation, it is difficult to predict the coupling wave paths beforehand,
and hence there must be uncertainties in the paths.
Under this uncertainty, the nominal controller may lead to deterioration of cancelation performance,
and even worse, it may make the feedback system unstable.
To solve this problem, we propose
\emph{robust} controller design against the uncertainty.

Let us assume the characteristic of the coupling wave paths
in Fig.~\ref{Relay Station Model} is perturbed as
\begin{Meqnarray}
 r e^{-Ls} \mapsto r e^{-Ls} + \displaystyle{ \sum_{i=1}^{M} r_{i} e^{-L_is}},
\end{Meqnarray}
where $r_i$ and $L_i$ are the attenuation ratio and the delay time of the $i$-th path,
respectively.
Note that $M$ represents the number of additional paths.
Since the additional paths are detour paths, we assume
\begin{Meqnarray}
 L_i > L, \quad i=1,2,\ldots,M.
\end{Meqnarray}
Then the characteristic of the feedback path in Fig.~\ref{Feedback Canceler} is perturbed as
\begin{Meqnarray}
 \alpha e^{-Ls}A_L \mapsto \alpha e^{-Ls}A_L + \sum_{i=1}^{M}\alpha_i e^{-L_is}A_{L_i}
\end{Meqnarray}
where $\alpha_i:=a_1a_2r_i$.
Define the error transfer function matrix $E(s)$ as
\begin{Meqnarray}
 E(s) := \sum_{i=1}^M \frac{\alpha_i}{\alpha}e^{-(L_i-L)s}A_{L_i-L}.
 \label{eq:Es}
\end{Meqnarray}
Since $A_L$ is a rotation matrix on $\mathbf{R}^2$ and the angle is $2\pi fL$ clockwise,
we have
\begin{Meqnarray}
  & & \alpha e^{-Ls}A_L + \sum_{i=1}^{M} \alpha_i e^{-L_is}A_{L_i} \nonumber\\
 & & \quad = \alpha e^{-Ls}A_L\biggl(I + \sum_{i=1}^{M}\dfrac{\alpha_i}{\alpha} e^{-(L_i-L)s} A_L^{-1} A_{L_i}\biggr)\nonumber \\
% &\quad& \alpha e^{-Ls}A_L\left(I + \sum_{i=1}^{M}\dfrac{\alpha_i}{\alpha} e^{-(L_i-L)s} A_{L_i-L}\right)\nonumber \\
 & & \quad = \alpha e^{-Ls}A_L\left(I + E(s)\right).
\end{Meqnarray}

Take a frequency weighting function matrix $W_2(s)$ that is real rational and satisfies
\begin{Meqnarray}
 \overline{\sigma}(E(j\omega)) < \overline{\sigma}(W_2(j\omega)), 
\end{Meqnarray}
for all $\omega \in {\mathbf R}$.
Since $A_{L_i-L}$ is an orthogonal matrix, the equation (\ref{eq:Es}) gives
\begin{Meqnarray}
 \label{singular value}
 & &\overline{\sigma}(E(j\omega))  \leq \sum_{i=1}^{M} \frac{\alpha_{i}}{\alpha} \overline{\sigma}\left( e^{-j(L_i-L)\omega}A_{L_i-L}\right) \leq \sum_{i=1}^M\frac{r_i}{r}.\nonumber\\
\end{Meqnarray}

Then the uncertainty in the coupling wave paths can be modeled as multiplicative perturbation, that is,
for any
$M>0$, $r_i \geq 0$, $L_i > L$ ($i=1,\ldots,M$), we have
\begin{Meqnarray}
 & &\alpha e^{-Ls}A_L + \displaystyle{ \sum_{i=1}^{M} \alpha_i e^{-L_is}A_{L_i}}\nonumber \\
 & & \quad \in \{\alpha e^{-Ls}A_L\left(1+\Delta(s)W_2(s)\right): \|\Delta\|_\infty < 1\}.
\end{Meqnarray}
From the inequality (\ref{singular value}), we can take
\begin{Meqnarray}
 W_2(s) = \left(\sum_{i=1}^M \frac{r_i}{r} + \varepsilon \right)I
 \label{eq:W2s}
\end{Meqnarray}
where $\varepsilon$ is an appropriately small and positive number.

Based on the formulation of the uncertainty, we consider the block diagram
shown in Fig.~\ref{Robust Model},
where $W_1(s)$ plays the same role as $W(s)$ in the nominal controller design.
\begin{figure}[t]
\includegraphics[width = 85mm]{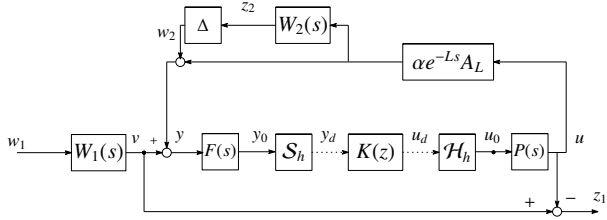}
\caption{Relay Station Model with Perturbation}
\label{Robust Model}
\end{figure}
Let $T_{z_1w_1}$ be the system from $w_1$ to $z_1$ and $T_{z_2w_2}$ be the system from $w_2$ to $z_2$.
If $\|T_{z_1w_1}\|_{\infty}$ is finite and $\|T_{z_2w_2}\|_{\infty} \leq 1$,
then the feedback system is robustly stable, that is, the feedback system is internally stable for
all $\Delta$ satisfying $\|\Delta\|_\infty < 1$
from the small gain theorem for sampled-data control systems \cite{Siv}.

Now we formulate the robust controller design problem as follows:
\begin{problem}
Find the digital controller (canceler) $K(z)$ that minimizes 
$\|T_{z_1w_1}\|_\infty$ subject to $\|T_{z_2w_2}\|_\infty \leq 1$.
\end{problem}

%The problem belongs to the NP-hard class generally and is unsolvable in practice \cite{Tok}.
To solve this problem, we adopt
the finite dimensional $Q$-parametrization where we limit feasible controllers \cite{Hindi}.
Then, the constraints are represented by linear matrix inequalities (LMI's) and the problem 
can be efficiently solved via numerical optimization software such as \texttt{SDPT3} or \texttt{SeDuMi} on \texttt{MATLAB} \cite{Toh,Str}.
For more details, see \cite{Hindi}.

%
% Section 4
%
\section{Design Examples}
\label{sec:sim}

In this section, we show simulation results to illustrate the effectiveness of the proposed methods.

We assume that the sampling period $h$ is normalized to $1$,
the carrier frequency $f$ is $10000$~Hz,
the attenuation rate of the coupling wave channel
$r = 0.2$,
and the time delay $L=1$.
Note that sampling frequency is $1$~Hz, which is much smaller than the carrier frequency.
Note also that the time delay is equal to the sampling period $h$.
We assume the low noise amplifier $a_1 = 1$.
An anti-alias analog filter is not employed in this examples, namely we assume $F(s) = I$.
The post filter $P(s)$ is modeled by
\begin{Meqnarray}
P(s) = \dfrac{1}{0.001s+1}I.
\end{Meqnarray}
We also assume the transmission gain to be
\begin{Meqnarray}
 a_2 = 1000,
\end{Meqnarray}
that is, 60~dB.
The frequency characteristic $W(s)$ is modeled by
\begin{Meqnarray}
W(s) = \dfrac{1}{2s+1}I. \label{eq:Ws}
\end{Meqnarray}
With these parameters, we compute the $H^\infty$-optimal nominal controller $K(z)$
by FSFH with discretization number $N=16$.

With this controller, we simulate coupling wave cancelation
with a random rectangular wave input with period $4$~s
filtered by the low-pass filter $P(s)$.
Note that this signal contains frequency components beyond the Nyquist frequency, 
$\pi/h = \pi$ [rad/sec], although the frequency of the wave, 
$\pi/8h = \pi/8$ [rad/sec] is much lower than $\pi$.
Fig.\ref{fb} shows the reconstructed signal $u$ in the feedback system (see Fig.~\ref{Feedback Canceler}).
The feedback system is guaranteed to be stable and 
the canceler achieves small reconstruction errors as shown in Fig.~\ref{fb_error}.
\begin{figure}[t]
\includegraphics[width = 85mm]{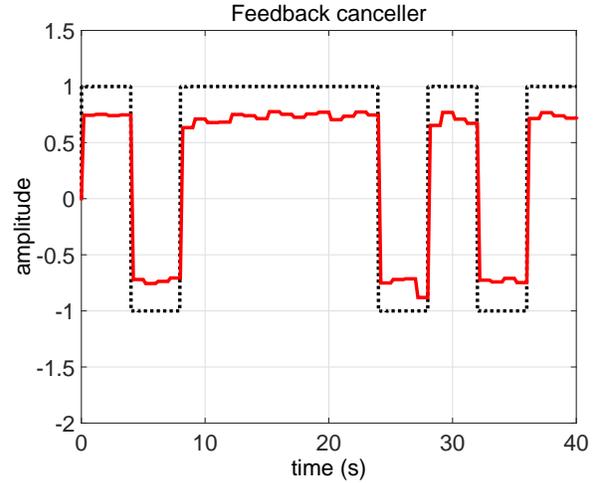}
\caption{Feedback Cancelation: input signal (dash-dot line), reconstructed signal $u$ by feedback canceler (solid line)}
\label{fb}
\end{figure}
\begin{figure}[t]
\includegraphics[width = 85mm]{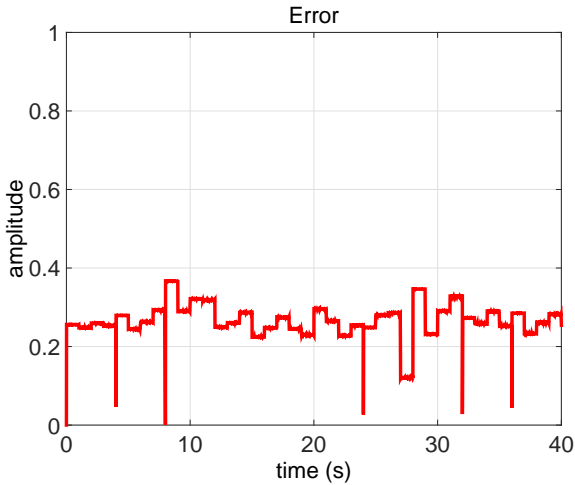}
\caption{Coupling wave effect $|v(t) - u(t)|$ by feedback canceler shown in Fig.~\ref{Feedback Canceler}}
\label{fb_error}
\end{figure}

Next, let us consider uncertainty in the coupling wave paths.
If the characteristic of the coupling wave paths changes, 
then the nominal controller may not work well.
To see this, let us design the nominal controller with
\begin{Meqnarray}
 a_2 = 100,
\end{Meqnarray}
that is, 40~dB,
and the other parameters are the same as above.
Then perturbing the nominal coupling wave path to be
\begin{Meqnarray}
 r e^{-Ls} + r_1 e^{-L_1s}
\end{Meqnarray}
where $r_1 = 0.07r, L_1 = 1.1L$.
It results in instability as shown in 
Fig.~\ref{fb_500}.
\begin{figure}[t]
\includegraphics[width = 85mm]{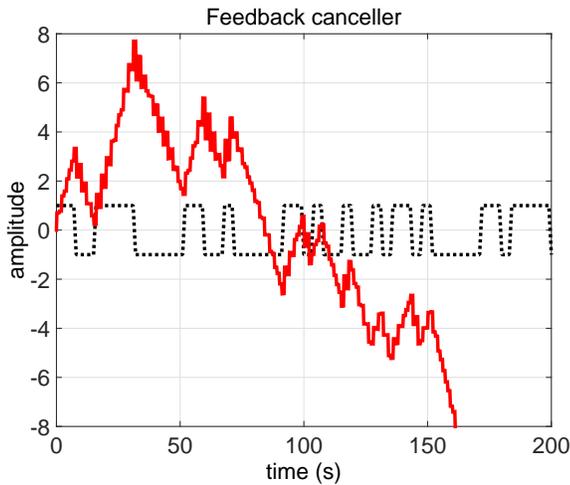}
\caption{Feedback Cancelation: input signal (dash-dot line) and reconstructed signal (solid line) 
with the perturbation}
\label{fb_500}
\end{figure}

To overcome this, we use the robust controller proposed in subsection \ref{sec:rob},
$M=1$, $r_1=0.1r$ with $W_1(s)=W(s)$ given in (\ref{eq:Ws}) and
$W_2(s)$ as in (\ref{eq:W2s}).
The dimension of $Q$-parametrization is $8$ with the FSFH number $N=4$.
Fig.~\ref{rob_500} shows the reconstructed signal, by which
we can observe the robust controller works well.

\begin{figure}[t]
\includegraphics[width = 85mm]{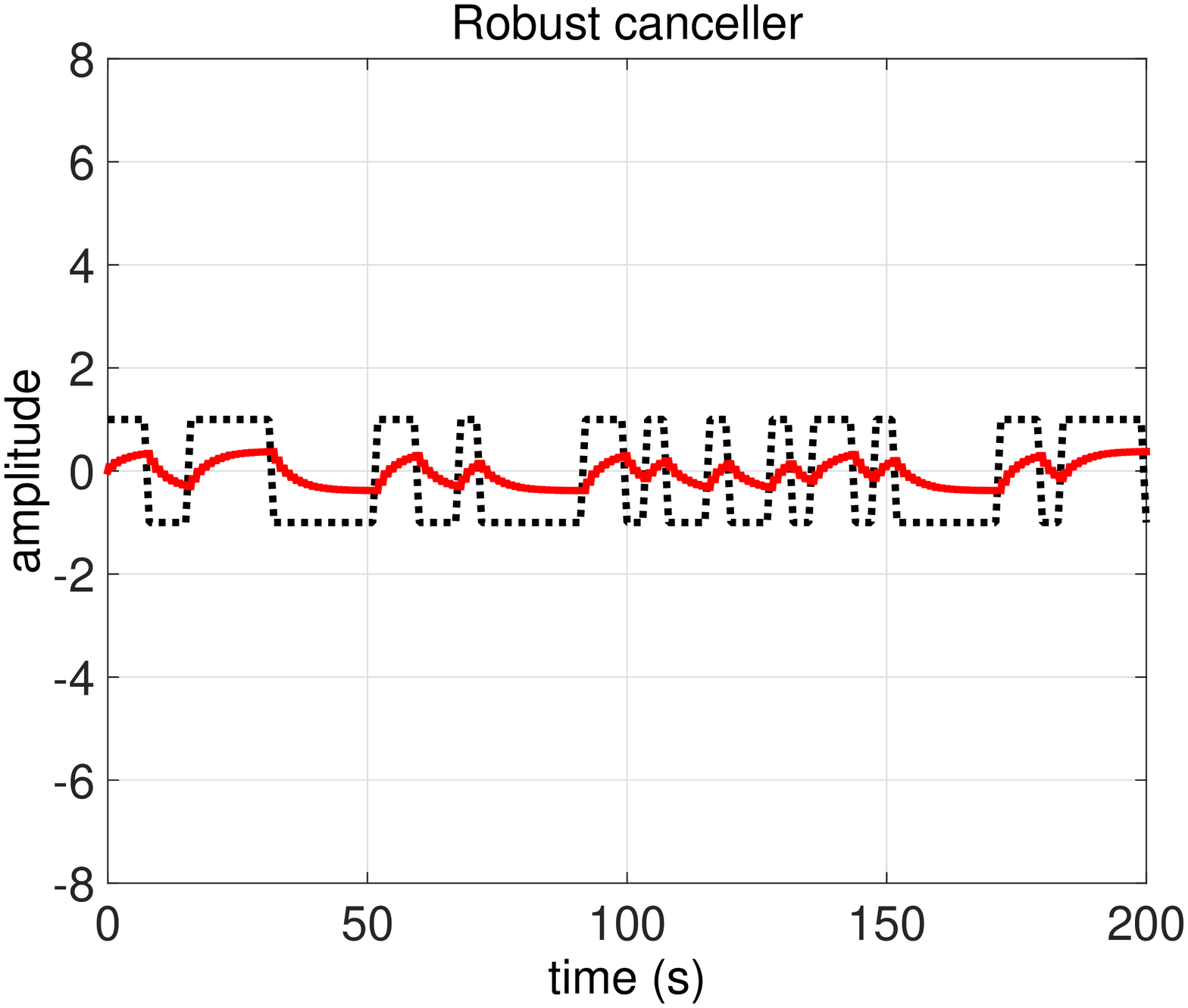}
\caption{Robust Cancelation: input signal (dash-dot line) and reconstructed signal (solid line) 
with the perturbation.}
\label{rob_500}
\end{figure}

%
% Section 5
%
\section{Conclusions}
\label{sec:conc}
In this paper, we have proposed feedback controller design
for self-interference cancelation in single-frequency full-duplex relay stations
based on the sampled-data $H^\infty$ control theory.
In particular, we proposed robust controller design against the unknown additive multipath.
Simulation results have been shown to illustrates the effectiveness of the proposed cancelers
in view of stability and robust stability.
Future work may include FIR (Finite Impulse Response) filter design and adaptive FIR filtering
as discussed in \cite{Nagahara13-1,Nag3}.

%\bibliographystyle{bst/jcmsisshr}
%\bibliography{sshrrefs}

\begin{biography}
\profile{s}{Hampei Sasahara}{
  He received his B.S. degrees from Kyoto University, Japan, in 2014.  
  He is currently a M.S student at Kyoto University.
  His research interests include signal processing and control theory.  
  He is a student member of IEEE and ISCIE.}  
  
\profile{m}{Masaaki Nagahara}{
He received the Bachelor's degree in engineering from Kobe University in 1998, 
the Master's degree and the Doctoral degree in informatics from Kyoto University 
in 2000 and 2003. He is currently a Senior Lecturer at Graduate School of Informatics, 
Kyoto University. His research interests include digital signal processing,
networked control systems. He received Young Authors Award in 1999
and Best Paper Award in 2012 from SICE,
Transition to Practice Award from IEEE Control Systems Society in 2012,
and Best Tutorial Paper Award from IEICE Communications Society in 2014.
He is a senior member of IEEE, and a member of ISCIE and IEICE.}  

\profile{n}{Kazunori Hayashi}{
He received the B. Eng., M. Eng., and Ph.D. degrees in communication engineering 
from Osaka University, Osaka, Japan, in 1997, 1999 and 2002, respectively. 
Since 2002, he has been with the Department of Systems Science, Graduate School 
of Informatics, Kyoto University. 
He is currently an Associate Professor there. 
His research interests include digital signal processing for communication systems.
He received the ICF Research Award from the KDDI Foundation in 2008, 
the IEEE Globecom 2009 Best Paper Award, the IEICE Communications Society 
Excellent Paper Award in 2011, the WPMC’11 Best Paper Award, 
the Telecommunications Advancement Foundation Award in 2012,
and the IEICE Communications Society Best Tutorial Paper Award from in 2014.
He is a member of IEEE, IEICE and ISCIE.
  }

\profile{f}{Yutaka Yamamoto}{
He received the Ph.D. degree in mathematics from the University
of Florida, Gainesville, in 1978, under the guidance
of Professor Rudolf E. Kalman. He is currently a
Professor in the Department of Applied Analysis
and Complex Dynamical Systems, Graduate School
of Informatics, Kyoto University, Kyoto, Japan.
His current research interests include the theory of sampled
data control systems, its application to digital
signal processing, realization and robust control of
distributed parameter systems and repetitive control.
Dr. Yamamoto was the recipient of G. S. Axelby Outstanding Paper Award
of the IEEE Control Systems Society (CSS) in 1996, the Commendation for Science
and Technology by the Minister of Education, Culture, Sports, Science and
Technology, Prize for Science and Technology, in Research Category in 2007,
Distinguished Member Award of the Control Systems Society of the IEEE in 2009,
Transition to Practice Award of the IEEE CSS in 2012,
and several other awards from the Society of Instrument and Control Engineers
(SICE) and the Institute of Systems, Control and Information Engineers
(ISCIE).   
He has served as Senior Editor of the IEEE Trans.\ Automatic Control
(TAC) for 2010--2011, and as an associate editor for several journals including TAC,
Automatica, Mathematics of Control, Signals and Systems, Systems and
Control Letters.  He was Chair of the Steering Committee of
MTNS for 2006--2008, and the General Chair of the MTNS 2006 in Kyoto.
He is currently Past President of the Control Systems Society (CSS)
of the IEEE, and
was President for 2013, President-Elect for 2012, 
a vice president for 2005--2008 of CSS, and President of ISCIE of
Japan for 2008--2009.
He is a fellow of the IEEE and SICE.}

%% \profile{m/s/f/h/n}{First~Middle Last}{Biography.}
%%   m; Member
%%   s: Student Member
%%   f: Fellow
%%   h: Honorary Member, Fellow
%%   n: Nonmember (Anything appears.) 
%% The photo file names should be 
%% a1.eps for the first author, a2.eps for the second author, etc.
\end{biography}

\end{document}